\documentclass{amsart}
     
\newtheorem{theorem}{Theorem}[section]
\newtheorem{lemma}[theorem]{Lemma}

\theoremstyle{definition}
\newtheorem{definition}[theorem]{Definition}

\theoremstyle{remark}
\newtheorem{remark}[theorem]{Remark}
 \usepackage{epsfig} 
 \usepackage{graphics} 
\numberwithin{equation}{section}

\def\a{{\alpha}}

\def\d{{\delta}}

\def\ep{{\eta}}
\def\z{{\zeta}}
\def\k{{\kappa}}
\def\l{{\lambda}}
\def\s{{\sigma}}

\def\f{{\phi}}
\def\na{{\nabla}}

\def\l{{\lambda}}
\def\m{{\mu}}

\def\p{{\psi}}

\def\G{{\Gamma}}
\def\O{{\Omega}}
\def\D{{\Delta}}

\def\ee{{{\mathcal E}}}

\def\mm{{{\mathcal M}}}

\def\ccc{{{\bf c}}}
\def\ddd{{{\bf d}}}
\def\bbb{{{\bf b}}}
\def\aaa{{{\bf a}}}
\def\xxx{{{\bf x}}}
\def\yyy{{{\bf y}}}
\def\nnn{{{\bf n}}}
\def\JJJ{{{\bf J}}}

\def\RRR{{{\bf R}^3}}
\def\RRRR{{{\bf R}^N}}

\def\9{{\ \hbox{in}\ \O}}
\def\1{{\ \hbox{on}\ \G}}
\def\2{{\ \hbox{on}\ \G_2}}
\def\3{{\ \hbox{on}\ \G_3}}
\def\0{{\ \hbox{on}\ \G}}

\def\pa{{\partial}}
\def\pp{{\parallel}}

\begin{document}

\title[The Maxwell-Volterra reciprocity law in electrodynamics]{The Maxwell-Volterra reciprocity law in electrodynamics}
\author{Giovanni Cimatti}
\address{Department of Mathematics, Largo Bruno
  Pontecorvo 5, 56127 Pisa Italy}
\email{cimatti@dm.unipi.it}






\begin{abstract}
Use is made of the theory of elliptic equations with measures data to prove the Maxwell-Volterra reciprocity law. A simple one-dimensional example is also given.
\end{abstract}

\maketitle

\section{Introduction}
In 1882 the Italian mathematician Vito Volterra, who was at that time a 23 years old student, published \cite{V} a paper entitled ``Sopra una legge di reciprocit\'a nella distribuzione delle temperature e delle correnti galvaniche costanti in un corpo qualunque''\footnote{In the same year Volterra published other 3 papers \cite{V1}, \cite{V2}, \cite{V3}.}.
The result was summarized by the author as follows ``se in un conduttore a due o tre dimensioni in cui la conducibilit\'a varia con continuit\'a da punto a punto si fa passare una corrente di intensit\'a $I$ da un punto $a$ a un punto $b$ e in due punti $c$ e $d$ si ha una differenza di potenziale, otterremo la stessa differenza fra i potenziali dei punti $a$ e $b$ quando si faccia entrare da $c$ e uscire da $d$ la stessa corrente di intensit\'a $I$''. 

Volterra was inspired in this research by the following passage of the book of J. C. Maxwell ``A Treatise on Electricity and Magnetism'' \cite{M} (Vol.1 page 405 second edition) : ``Let $A$, $B$,  $C$, $D$ be any four points of a linear system of conductor and let the effect of a current $Q$, made to enter at $A$ and leave it at $B$, be to make the potential at $C$ exceed that at $D$ by $P$. Then, if an equal current $Q$ be made to enter the system at $C$ and leave it at $D$, the potential at $A$ will exceed that at $B$ by the same quantity P''. The law of Maxwell refers to linear conductors and the intention of Volterra was to generalize Maxwell's result to plane or three-dimensional conductors.

 Volterra's proof, although clever, is not fully rigorous. This is understandable, since to model correctly the injection of a current in a single point of a conductor one needs the theory of distribution which was developed much later. Goal of this paper is to show that  the reciprocity result of Volterra is a simple consequence of the theory of elliptic equations with measures in the right hand side developed mainly by Guido Stampacchia in the 60' of the last century. We refer in particular to the papers \cite{LWS}, \cite{S1} and, for review papers, to \cite{B} and  \cite{O}. Those results are summarized in Sections 3 and 4. Moreover, we show that, if the conductor is not homogeneous or not isotropic and therefore the electrical conductivity is represented by a symmetric tensor\footnote{The symmetry of the conductivity tensor is a consequence of the Onsager's principle, see \cite{L}, page 148.} $a_{ij}(\xxx)$ the Volterra's reciprocity law is true only if $a_{ij}(\xxx)$ satisfies certain compatibility conditions. This situation occurs quite often. As an example, see \cite{H}, let us consider a conductor of heat and electricity $\O$ with boundary $\pa\O$. Denote by $\f$ the electric potential, by $u$ the temperature and by $\s(u)$ and $k(u)$ the electric conductivity and the thermal conductivity, both depending on the temperature. Neglecting Joule heating we arrive, under steady conditions, to the following boundary value problem

\begin{equation}
\label{a}
\na\cdot(\s(u)\na\f)=0\ \ \hbox{in}\ \O,\quad \f=\f_b\ \ \hbox{on}\ \pa\O
\end{equation}

\begin{equation}
\label{b}
\na\cdot(\k(u)\na u)=0\ \ \hbox{in}\ \O,\quad\ u=u_b\quad\hbox{on}\ \pa\O,
\end{equation}
where $u_b$ denotes the prescribed temperature on the boundary and $\f_b$ the given potential on $\pa\O$. From (\ref{b}) using the Kirchhoff's transformation the temperature $u(x)$ can be obtained. Substituting $u=u(\xxx)$ in (\ref{a}) we find a virtual ``electrical conductivity'' $a(\xxx)=\s(u(\xxx))$ depending on $\xxx=(x_1,x_2,x_3)\in\O$.

 We point out that Volterra supposed the conductivity to be continuous whereas in the proof we present here it is enough to assume the conductivity tensor in $L^\infty(\O)$ and to satisfy an uniform ellipticity condition.

\section{Derivation of the basic equations. A one-dimensional case.}
Let $\O$ be a bounded, open and connected subset of $\RRR$ with a regular boundary $\pa\O$. $\O$ represents a solid conductor of electricity. Let $a_{ij}(\xxx)$ be a symmetric tensor giving the electric conductivity. Let $\f(\xxx)$ be the electric potential in $\O$. We assume the boundary of $\O$ to be grounded. This corresponds to the boundary condition $\f=0\ \  \hbox{on}\ \pa\O$. Let $\aaa$ and $\bbb$ be two distinct points of $\O$. We assume that a current $I$ is injected in $\aaa$ and the same current is extracted from $\bbb$. To model this situation we recall that the current density $\JJJ$ in $\O$ is given by the local form of Ohm's law i.e.\footnote{Here and hereafter use is made of the summation convention.}

\begin{equation*}
J_i=a_{ij}\f_{x_j}.
\end{equation*}
Let $B_\aaa$ and $B_\bbb$ be two open spheres with centers $\aaa$ and $\bbb$ and equal radius $L$ such that $\bar B_\aaa\cap\bar B_\bbb=\emptyset$, $\bar B_\aaa\cap\pa\O=\emptyset$ and $\bar B_\bbb\cap\pa\O=\emptyset$. If $\nnn^{(\aaa)}$ denotes the exterior pointing unit normal vector to $\partial B_\aaa$, the total current crossing $\partial B_\aaa$ is given by

\begin{equation*}
\label{2_8}
\int_{\partial B_\aaa} a_{ij}(\xxx)\f_{x_i}n^{(\aaa)}_j d\sigma.
\end{equation*}
Let $\z^{(\aaa)}(\xxx)$ be the solution, vanishing on $\pa\O$, of the equation

\begin{equation}
\label{3_8}
\bigl(a_{ij}(\xxx)\z^{(\aaa)}_{x_i}\bigl)_{x_j}=\d_{\aaa}(\xxx),
\end{equation}
where $\d_{\aaa}$ denotes the Dirac's measure pointed in $\aaa$. Define the number

\begin{equation*}
\a^{(\aaa)}=\int_{\partial B_\aaa} a_{ij}(\xxx)\z_{x_i}^{(\aaa)}n^{(\aaa)}_j d\sigma.
\end{equation*}
Similarly, let $\z^{(\bbb)}(\xxx)$ be the solution of the equation

\begin{equation}
\label{1_10}
\bigl(a_{ij}(\xxx)\z^{(\bbb)}_{x_i}\bigl)_{x_j}=\d_{\bbb}(\xxx)
\end{equation}
and

\begin{equation*}
\a^{(\bbb)}=\int_{\partial B_\bbb} a_{ij}(\xxx)\z_{x_i}^{(\bbb)}n^{(\bbb)}_j d\sigma.
\end{equation*}
Assume that the numerical value of the current injected in $\aaa$ and extracted in $\bbb$ is $I$ and define

\begin{equation}
\label{3_10}
\f(\xxx)=\frac{I}{\a^{(\aaa)}}\z^{(\aaa)}(\xxx)-\frac{I}{\a^{(\bbb)}}\z^{(\bbb)}(\xxx).
\end{equation}
We claim that $\f(\xxx)$ is the electric potential corresponding  to the described injection and extraction of currents. In fact,

\begin{equation}
\label{1_11}
\int_{\pa B_\aaa}a_{ij}\f_{x_i}n_j^{(\aaa)} d\s=\frac{I}{\a^{(\aaa)}}\int_{\pa B_\aaa}a_{ij}\z_{x_i}^{(\aaa)}n_j^{(\aaa)}d\s-\frac{I}{\a^{(\bbb)}}\int_{\pa B_\bbb}a_{ij}\z_{x_i}^{(\bbb)}n_j^{(\bbb)}d\s.
\end{equation}
On the other hand, $\z^{(\bbb)}$ satisfies in $\bar B_{\aaa}$ the equation

\begin{equation*}
\bigl(a_{ij}\z_{x_i}^{(\bbb)}\bigl)_{x_j}=0.
\end{equation*}
Therefore, the second term in the right hand side of (\ref{1_11}) vanishes by the divergence theorem. Hence we have

\begin{equation*}
\int_{\pa B_\aaa}a_{ij}\f_{x_i}n_j^{(\aaa)} d\s=I.
\end{equation*}
In a similar way we obtain, by (\ref{3_8}), (\ref{1_10}) and (\ref{3_10}),

\begin{equation*}
\int_{\pa B_\bbb}a_{ij}\f_{x_i}n_j^{(\bbb)} d\s=-I.
\end{equation*}
We conclude that the potential $\f(\xxx)$ satisfies the equation

\begin{equation*}
\bigl(a_{ij}\f_{x_i}\bigl)_{x_j}=\frac{I}{\a^{(\aaa)}}\d_{\aaa}-\frac{I}{\a^{(\bbb)}}\d_{\bbb}.
\end{equation*}
The special case

\begin{equation}
\label{2_12}
 a_{ij}=\d_{ij}
\end{equation}
is particularly simple. For, if (\ref{2_12}) holds $\z^{(\aaa)}(\xxx)$ is the solution of the equation

\begin{equation*}
 \D\z^{(\aaa)}=\d_\aaa(\xxx)
\end{equation*}
vanishing on $\pa\O$. Let 

\begin{equation*}
 \ee(\xxx)=-\frac{1}{4\pi|\xxx|}.
\end{equation*}
$\ee(\xxx)$ is the fundamental solution of the Laplace operator in three dimension. We have

\begin{equation*}
 \D\ee=\d_0(\xxx).
\end{equation*}
Thus

\begin{equation}
\label{2_13}
 \ep^{(\aaa)}(\xxx)=\z^{(\aaa)}(\xxx)-\ee(\xxx-\aaa)
\end{equation}
is the regular part of $\z^{(\aaa)}$ and 

\begin{equation*}
 \D\ep^{(\aaa)}=0\quad \9.
\end{equation*}
Hence we have

\begin{equation*}
 \int_{\pa B_\aaa}\ep_{x_i}^{(\aaa)}n_i^{(\aaa)}d\s=0.
\end{equation*}
By (\ref{2_13}) we obtain $\z^{(\aaa)}=\ep^{(\aaa)}(\xxx)+\ee(\xxx-\aaa)$. Thus

\begin{equation*}
 \a^{(\aaa)}=\int_{\pa B_\aaa}\z^{\aaa}_{x_i}n_i^{(\aaa)}d\s=\int_{\pa B_\aaa}\ee_{x_i}(\xxx-\aaa)n_i^{(\aaa)}d\s=1.
\end{equation*}
In the same way we find $\a^{(\bbb)}=1$.
\vskip .2cm

The Volterra's reciprocity law is stated in the following 

\begin{theorem}

Let $a_{ij}(\xxx)\in L^\infty(\O)$ and 

\begin{equation}
a_{ij}=a_{ji},\quad a_{ij}\xi_i\xi_j\geq \l |\xi|^2\quad \hbox{for all}\quad \xi\in\RRR.
\label{1_15}
\end{equation}

 Let $\aaa$, $\bbb$, $\ccc$ and $\ddd$ be four distinct, but otherwise arbitrary, points of $\O$ considered in this order. Let $\f^{(1)}(\xxx)$ be the unique solution of the equation

\begin{equation}
\label{1_16}
\bigl(a_{ij}\f_{x_i}^{(1)}\bigl)_{x_j}=\frac{I}{\a^{(\aaa)}}\d_{\aaa}-\frac{I}{\a^{(\bbb)}}\d_{\bbb}
\end{equation}
vanishing on $\pa\O$ and $\f^{(2)}(\xxx)$ the unique solution of the equation

\begin{equation}
\label{2_16}
\bigl(a_{ij}\f_{x_i}^{(2)}\bigl)_{x_j}=\frac{I}{\a^{(\ccc)}}\d_{\ccc}-\frac{I}{\a^{(\ddd)}}\d_{\ddd}
\end{equation}
vanishing on $\pa\O$. Then we have

\begin{equation}
\label{1_17}
\f^{(2)}(\aaa)-\f^{(2)}(\bbb)-\bigl[\f^{(1)}(\ccc)-\f^{(1)}(\ddd)\bigl]=I\bigl(\frac{1}{\a^{(\ccc)}}-\frac{1}{\a^{(\aaa)}}\bigl)g(\aaa,\ccc)-I\bigl(\frac{1}{\a^{(\ddd)}}-\frac{1}{\a^{(\bbb)}}\bigl)g(\bbb,\ddd),
\end{equation}
where $g(\xxx,\yyy)$ is the Green's function of the operator $L\f=\bigl(a_{ij}\f_{x_i}\bigl)_{x_j}$ corresponding to homogeneous boundary  conditions.
In particular, if

\begin{equation}
\label{2_17}
\a^{(\ccc)}=\a^{(\aaa)}\quad \hbox{and}\quad \a^{(\ddd)}=\a^{(\bbb)},
\end{equation}
we have

\begin{equation}
\label{3_17}
\f^{(2)}(\aaa)-\f^{(2)}(\bbb)=\f^{(1)}(\ccc)-\f^{(1)}(\ddd).
\end{equation}
\end{theorem}

The proof of Theorem 2.1 is given in Section 5, the existence and uniqueness of $\f^{(1)}$ and $\f^{(2)}$ is proven in Section 3. The definition of $g(\xxx,\yyy)$, together with its properties is given in Section 4.
\vskip .3 cm

The Volterra's reciprocity is easily verified in the one-dimensional case. Let $\O=(0,L)$ and $a$, $b$, $c$ and $d$ be four positive numbers less than $L$ considered in this order. Assume a constant conductivity equal to $1$. To determine $\f(x)$ we have to solve the two-point problem

\begin{equation}
\f''(x)=I\d_a(x)-I\d_b(x),\quad \f(0)=0,\quad \f(L)=0,
\label{1_18}
\end{equation}
where $I$ is the total current injected in $a$ and extracted in $b$. The unique solution of problem (\ref{1_18}) is 

\begin{equation}
\label{3_18}
\f(x;a,b)=I\Bigl[(x-a)H(x-a)-(x-b)H(x-b)+\frac{1}{L}x(a-b)\Bigl],
\end{equation}
where $H(x)$ is the Heaviside step function\footnote{ $H(x)=1$ if $x>0$, $H(x)=0$ if $x<0$.}. To verify directly (\ref{3_18}) we recall that $x\d'_0(x)=-\d_0(x)$. We have from (\ref{3_18})

\begin{equation*}
\f(d,a,b)-\f(c,a,b)-[\f(b,c,d)-\f(a,c,d)]=I(d-a)H(d-a)-I(d-b)H(d-b)+\frac{I}{L}d(a-b)
\end{equation*}

\begin{equation*}
-I(c-a)H(c-a)+I(c-b)H(c-b)-\frac{I}{L}c(a-b)-\biggl[I(b-c)H(b-c)-I(b-d)H(b-d)+\frac{I}{L}b(c-d)
\end{equation*}

\begin{equation*}
-I(a-c)H(a-c)+I(a-d)H(a-d)-\frac{I}{L}a(c-d)\biggl]=0
\end{equation*}
i.e. the Volterra's reciprocity law.

\section{Results of the theory of elliptic boundary value problems with measures in the right hand side}
The topics presented in this Section are mainly taken from \cite{LWS}, \cite{S1}, \cite{B} and \cite{O}. Since its purely mathematical character we assume here $\O\subset\RRRR$, $N\geq 3$ with a regular boundary $\pa\O$. Let us consider the operator

\begin{equation*}
Lu=-\bigl(a_{ij}(x)u_{x_i}\bigl)_{x_j},
\end{equation*}
where $a_{ij}(\xxx)=a_{ji}(\xxx)$ are bounded measurable functions which satisfy

\begin{equation*}
\l^{-1}|\xi|^2\leq a_{ij}\xi_i\xi_j\leq\l|\xi|^2,\quad \l>0.
\end{equation*}
We recall first of all two well-known results \cite{MS}, \cite{Lady}.

\begin{theorem}
Given $f^{(i)}\in L^2(\O)$, $i=0,1,..,N$ there exists one and only one solution of the problem
\begin{equation}
\label{1_26}
u\in H^1_0(\O),\quad Lu=f^{(0)}+ f^{(i)}_{x_i}.
\end{equation}
\end{theorem}
\vskip .5 cm
\begin{theorem}
If  $f^{(i)}\in L^p(\O)$, $i=0,...,N$ with $p>N$ and $\pa\O$ is of class $C^2$ the solution of (\ref{1_26}) is Hoelder continuous in $\bar\O$ and

\begin{equation*}
\max_{\bar\O}|u(\xxx)|\leq C \sum_{i=0}^N \pp f^{(i)}\pp_{L^p(\O)}.
\end{equation*}
\end{theorem}
\noindent By Theorem 3.1 there is a continuous linear operator $G$ (the Green's operator) defined in $H^{-1,2}(\O)$ with values in $H^{1,2}_0(\O)$ such that $u(\xxx)=G(T)$ gives the unique solution in $H^{1,2}_0(\O)$ of the equation $Lu=T$ when $T\in H^{-1}(\O)$. 

We denote $\mm$ the space of measures of bounded variations on $\O$. We give below our basic 

\begin{definition}Let $\m\in\mm(\O)$. We say that $u(\xxx)\in L^1(\O)$ is a weak solution vanishing on $\pa\O$ of the equation $Lu=\m$ if $u(\xxx)$ satisfies

\begin{equation}
\label{2_28}
\int_\O uL\f\  dx=\int_\O \f d\m
\end{equation}
for all $\f\in H^{1,2}_0(\O)\cap C^0(\bar \O)$ such that $L\f\in C^0(\bar\O)$.
\end{definition}

\begin{remark}
Let $\p=L\f$ and $\f=G(\p)$. Since $L\f\in C^0(\bar\O)$ we can rewrite (\ref{2_28}) as follows

\begin{equation*}
u\in L^1(\O),\quad \int_\O u\p dx=\int_\O G(\p) d\m
\end{equation*}
for every $\p\in C^0(\bar\O)$. Or, in the ``mixed'' form,

\begin{equation*}
u\in L^1(\O),\quad \int_\O u\p dx=\int_\O \f d\m
\end{equation*}
for all $\p\in C^0(\bar\O)$ and $\f\in H^{1,2}_0(\O)\cap C^0(\bar\O)$.
\end{remark}

\begin{theorem}
Let $\m\in\mm$. There exists a unique weak solution of the equation $Lu=\m$ vanishing on $\pa\O$. Moreover,

\begin{equation*}
u\in H^{1,p}(\O)\quad \hbox{for every}\quad 0<p'<\frac{N}{N-1}.
\end{equation*}
\end{theorem}

\begin{proof}
Because $\p(\xxx)\in C^0(\bar\O)$ we have $\f(\xxx)\in H^{1,p}_0(\O)$ with $p>N$.
Therefore, the Green's operator $G:H^{-1,p}(\O)\to C^0(\bar\O)$ is one-to-one. Let us consider the adjoint $G^*$ of $G$. Since $G\bigl(H^{-1,p}(\O)\bigl)\subset C^0(\bar\O)$ and the dual space of $C^0(\bar\O)$ is $\mm(\O)$, $G^*$ is certainly defined for all $\m\in\mm(\O)$. The range of $G^*$ is the dual space of $H^{-1,p}(\O)$ i.e. $H^{1,p'}(\O)$ with $\frac{1}{p}+\frac{1}{p'}=1$. 
By definition we have

\begin{equation*}
\int_\O u\p dx=\int_\O G(\p)d\m
\end{equation*}
thus $u=G^*(\m)$ with $ u\in H^{1,p'}(\O)$ and  $0<p'<\frac{N}{N-1}$  is the unique solution of the equation $Lu=\m$ vanishing on $\pa\O$.


\end{proof}

\section{the Green's function}

\begin{definition}
The weak solution $g(\xxx,\yyy)$ vanishing on $\pa\O$ of the equation

 \begin{equation*}
Lg(\xxx,\yyy)=\d_\yyy(\xxx)
\end{equation*} 
is termed Green's function of the operator

\begin{equation}
\label{2_33}
Lu=-\bigl(a_{ij}u_{x_i}\bigl)_{x_j}.
\end{equation}
\end{definition}
\noindent In (\ref{2_33}) the coefficients $a_{ij}(\xxx)$ are assumed to satisfy

\begin{equation}
\label{3_33}
a_{ij}(\xxx)\in L^\infty(\O),\quad a_{ij}(\xxx)=a_{ji}(\xxx),\quad  a_{ij}(\xxx)\xi_i\xi_j\geq \l |\xi|^2,\ \l>0.
\end{equation}
Using the results of Section 3 we can prove the well-known

\begin{theorem}
Let

\begin{equation}
\label{1_34}
\p(\xxx)\in C^0(\bar\O).
\end{equation}
The solution of the boundary value problem

\begin{equation}
\label{2_34}
\f\in H^{1,2}_0(\O)\cap C^0{\bar\O)},\quad Lu=\p
\end{equation}
is given by

\begin{equation}
\label{3_34}
\f(\yyy)=\int_\O g(\xxx,\yyy)\p(\xxx)dx.
\end{equation}
\end{theorem}

\begin{proof}
By definition of weak solution we have

\begin{equation}
\label{4_34}
\int_\O g(\xxx,\yyy)\p(\xxx)dx=\int_\O\f(\xxx)d\d_\yyy(\xxx).
\end{equation}
 On the other hand, by the basic property of the Dirac's measure, we have

\begin{equation}
\label{5_34}
\int_\O \f(\xxx)d\d_\yyy(\xxx)=\f(\yyy).
\end{equation}
From (\ref{4_34}) and (\ref{5_34}) we obtain (\ref{3_34}).
\end{proof}

When the coefficients $a_{ij}(\xxx)$ of the operator $L$ are of class $C^1(\bar\O)$ the positivity

\begin{equation}
\label{1_35}
g(\xxx,\yyy)\geq 0
\end{equation}
and the symmetry

\begin{equation}
\label{2_35}
g(\xxx,\yyy)=g(\yyy,\xxx),\quad \xxx\neq\yyy,\quad \xxx,\yyy\in\O
\end{equation}
are classical properties of the Green's function (see \cite{K}, page 238). We want to show that (\ref{1_35}) and (\ref{2_35}) remain valid under the more general assumptions (\ref{3_33}). To this end we use the following approximation argument (see \cite{LWS} for the proof).

\begin{theorem}
 Let $a_{ij}^{(s)}(\xxx)$ satisfy

\begin{equation*}
a_{ij}^{(s)}=a_{ji}^{(s)}, \quad a_{ij}^{(s)}\in C^\infty(\bar\O),\quad a_{ij}^{(s)}(\xxx)\xi_i\xi_j\geq\l |\xi|^2,\quad \l>0
\end{equation*}
and

\begin{equation*}
a_{ij}^{(s)}\to a_{ij}\ \hbox{in}\ \ L^q(\O)\ \ \hbox{for all}\  0<q<\infty.
\end{equation*}
Define

\begin{equation}
\label{c}
 L^{(s)}u^{(s)}=-\bigl(a_{ij}^{(s)} u_{x_i}^{(s)}\bigl)_{x_j}.
\end{equation}
Then for any measure $\m\in\mm(\O)$ the weak solution $u^{(s)}$ of the equation

\begin{equation*}
L^{(s)}u^{(s)}=\m
\end{equation*}
vanishing on $\pa\O$ converges to the weak solution  $u$ vanishing on $\pa\O$ of the equation

\begin{equation*}
Lu=\m,\quad Lu=-\bigl(a_{ij} u_{x_i}\bigl)_{x_j}
\end{equation*}
weakly in $H^{1,p}_0(\O)$ for any $0<p<\frac{N}{N-1}$ and strongly in $L^q(\O)$ with $0<q<\frac{N}{N-2}$.
\end{theorem}
\bigskip
Let $g^{(s)}(\xxx,\yyy)$ be the Green's function corresponding to the operator $L^{(s)}$. According to the Definition 4.1 $g^{(s)}(\xxx,\yyy)$ is the weak solution of the equation

\begin{equation*}
L^{(s)}g^{(s)}(\xxx,\yyy)=\d_\yyy(\xxx)
\end{equation*}
vanishing on $\pa\O$. By Theorem 4.3 we have 

\begin{equation*}
g^{(s)}(\xxx,\yyy)\to g(\xxx,\yyy)\quad \hbox{for}\ s\to\infty
\end{equation*}
in $L^q(\O)$ with $0<q<\frac{N}{N-2}$. On the other hand, by the classical theory

\begin{equation*}
g^{(s)}(\xxx,\yyy)=g^{(s)}(\yyy,\xxx),\quad\ g^{(s)}(\xxx,\yyy)\geq 0,\quad\ \xxx\neq\yyy.
\end{equation*}
In the limit we have

\begin{equation*}
g(\xxx,\yyy)=g(\yyy,\xxx),\quad\ g(\xxx,\yyy)\geq 0,\quad\ \xxx\neq\yyy.
\end{equation*}
We now generalize the formula

\begin{equation*}
\f(\yyy)=\int_\O g(\xxx,\yyy)\p(\xxx)\ dx
\end{equation*}
which gives the solution of the problem

\begin{equation}
\label{2_42}
\f\in H^{1,2}(\O),\quad L\f=\p,\quad \p\in C^0(\bar\O)
\end{equation}
to the case when the right hand side in (\ref{2_42}) is a measure of bounded variation. Use will be made of the following

\begin{lemma}
Let $\m\in \mm$ be non-negative and let $g(\xxx,\yyy)$ be the Green's function of the operator $L$ with zero boundary condition. Then the integral

\begin{equation*}
\int_\O g(\xxx,\yyy)\ d\m(\yyy)
\end{equation*}
exists finite for almost every $\xxx\in\O$ and, as a function of $\xxx\in\O$, belongs to $L^1(\O)$.
\end{lemma}

\begin{proof}
Let $\z(\xxx)\in C^0(\bar\O)$, $\z(\xxx)\geq 0$ and consider the solution $\f(\xxx)$ of the problem

\begin{equation*}
\f\in H^{1,2}_0(\O),\quad L\f=\z.
\end{equation*}
By Theorem 3.2 we have $\f(\xxx)\in C^0(\bar\O)$ and by Theorem 4.2

\begin{equation}
\label{2_49}
\f(\yyy)=\int_\O g(\xxx,\yyy)\z(\xxx)dx.
\end{equation}
Since $\f(\yyy)$ is continuous, the integral 

\begin{equation}
\label{3_49}
\int_\O \f(\yyy)\ d\m(\yyy)
\end{equation}
exists finite. Moreover, the iterated integral 

\begin{equation*}
\int_\O\biggl[\int g(\xxx,\yyy)\z(\xxx)\ dx\biggl]\ d\m(\yyy)
\end{equation*}
which is obtained substituting (\ref{2_49}) into (\ref{3_49}) is convergent. On the other hand, $g(\xxx,\yyy)\geq 0$ thus $g(\xxx,\yyy)\z(\xxx)\m(\yyy)$ is a non-negative measure of bounded variation. By Fubini's theorem, (see \cite{F} page 88) the integral

\begin{equation*}
\int_{\O\times\O} g(\xxx,\yyy)\z(\xxx)\ dx\ d\m(\yyy)
\end{equation*}
exists finite. With the choice $\z=1$ we have

\begin{equation*}
\int_{\O\times\O} g(\xxx,\yyy)\ dx\ d\m(\yyy)<\infty.
\end{equation*}
Again by Fubini's theorem the sectional function

\begin{equation*}
u(\xxx)=\int_\O g(\xxx,\yyy)\ d\m(\yyy)
\end{equation*}
is well-defined and belongs to $L^1(\O)$.
\end{proof}
\vskip .3cm

\begin{theorem}
Let $\m$ be a measure of bounded variation and $g(\xxx,\yyy)$ the Green's function of  the operator $L$. Then 

\begin{equation}
\label{1_52}
u(\xxx)=\int_\O g(\xxx,\yyy)\ d\m(\yyy)
\end{equation}
gives the weak solution, vanishing on $\pa\O$, of the equation

\begin{equation}
\label{2_52}
L u=\m.
\end{equation}
\end{theorem}

\begin{proof}
We can write $\m=\m^+-\m^-$, with $\m^+$ and $\m^-$ both positive measures. On the other hand, the equation (\ref{2_52}) is linear and we have $u=u^+-u^-$ whereas $u^-$ and $u^+$ are the solutions, vanishing on $\pa\O$, of the equations $Lu^-=\m^-$ and $Lu^+=\m^+$. Hence  it is not restrictive to assume in the proof $\m\geq 0$. By Lemma 4.4 the right hand side of (\ref{1_52}) as a function of $\xxx\in\O$ is a well-defined function of $L^1(\O)$. Assume $\p(\xxx)\in C^0(\bar\O)$ and let $\f(\xxx)$ be the solution of the problem

\begin{equation*}
\f\in H^{1,2}_0(\O),\quad L\f=\p.
\end{equation*}
By Theorem 4.2 we have

\begin{equation}
\label{2_53}
\f(\yyy)=\int_\O g(\xxx,\yyy)\p(\xxx)\ dx.
\end{equation}
If $u(\xxx)$ is given by (\ref{1_52}), we obtain

\begin{equation}
\label{3_53}
\int_\O u(\xxx)\p(\xxx)\ dx=\int_\O\biggl[\int_\O g(\xxx,\yyy)d\m(\yyy)\biggl]\p(\xxx)\ dx,
\end{equation}
where in the right hand side we have an iterated integral. Applying the Fubini's theorem we can invert the order of integration in (\ref{3_53}). We obtain, using (\ref{2_53}),

\begin{equation*}
\int_\O\biggl[\int_\O g(\xxx,\yyy)d\m(\yyy)\biggl]\p(\xxx)\ dx=\int_\O\biggl[\int_\O g(\xxx,\yyy)\p(\xxx)\ dx\biggl]d\m(\yyy)=\int_\O\f(\yyy)d\m(\yyy).
\end{equation*}
Hence, from (\ref{3_53}) we get

\begin{equation*}
\int_\O u(\xxx)\p(\xxx)\ dx=\int_\O \f(\yyy)\ d\m(\yyy).
\end{equation*}
Thus $u(\xxx)$ is the weak solution of equation (\ref{2_52}) vanishing on $\pa\O$.
\end{proof}

\section{Proof of the Volterra's reciprocity Theorem 2.1}
By Theorem 3.5 there exists a unique weak solution $\f^{(1)}(\xxx)$ of the equation

\begin{equation}
\label{1_55}
\bigl(a_{ij}(\xxx)\f^{(1)}_{x_i}\bigl)_{x_j}=\frac{I}{\a^{(\aaa)}}\d_{\aaa}(\xxx)-\frac{I}{\a^{(\bbb)}}\d_{\bbb}(\xxx)
\end{equation}
vanishing on $\pa\O$. The same can be said for the equation

\begin{equation}
\label{2_55}
\bigl(a_{ij}(\xxx)\f^{(2)}_{x_i}\bigl)_{x_j}=\frac{I}{\a^{(\ccc)}}\d_{\ccc}(\xxx)-\frac{I}{\a^{(\ddd)}}\d_{\ddd}(\xxx).
\end{equation}
Let $g(\xxx,\yyy)$ be the Green's function of the operator $L$. Let $\aaa$, $\bbb$, $\ccc$ and $\ddd$ be arbitrary, distinct points of $\O$ considered in this order. By Theorem 4.5 we have 

\begin{equation}
\label{1_57}
\f^{(2)}(\aaa)-\f^{(2)}(\bbb)-\bigl[\f^{(1)}(\ccc)-\f^{(1)}(\ddd)\bigl]=
\end{equation}

\begin{equation*}
=\frac{I}{\a^{(\ccc)}}\int_\O g(\aaa,\yyy)\ d\d_\ccc(\yyy)-\frac{I}{\a^{(\ddd)}}\int_\O g(\bbb,\yyy)\ d\d_\ddd(\yyy)-\frac{I}{\a^{(\aaa)}}\int_\O g(\ccc,\yyy)\ d\d_\aaa(\yyy)+\frac{I}{\a^{(\bbb)}}\int_\O g(\ddd,\yyy)\ d\d_\bbb(\yyy).
\end{equation*}
Recalling the basic property of the Dirac's measure and the symmetry (\ref{2_35}) of the Green's function we obtain

\begin{equation*}
\f^{(2)}(\aaa)-\f^{(2)}(\bbb)-\biggl[\f^{1}(\ccc)-\f^{(1)}(\ddd)\biggl]=I\biggl(\frac{1}{\a^{(\ccc)}}-\frac{1}{\a^{(\aaa)}}\biggl)g(\aaa,\ccc)-I\biggl(\frac{1}{\a^{(\ddd)}}-\frac{1}{\a^{(\bbb}}\biggl)g(\bbb,\ddd).
\end{equation*}
Hence, if 

\begin{equation}
\label{1_58}
\frac{1}{\a^{(\ccc)}}=\frac{1}{\a^{(\ddd)}}\ \hbox{and}\ \frac{1}{\a^{(\ddd)}}=\frac{1}{\a^{(\bbb)}}
\end{equation}
we have

\begin{equation}
\label{2_58}
\f^{(2)}(\aaa)-\f^{(2)}(\bbb)=\f^{(1)}(\ccc)-\f^{(1)}(\ddd).
\end{equation}
The condition (\ref{1_58}) is satisfied if, in particular, $a_{ij}(\xxx)=\s\d_{ij}$ with $\s>0$ i.e. for an homogeneous conductor.

\begin{remark}
We note that the reciprocity condition (\ref{2_58}) may hold also in other cases. Let us consider, for example, instead of (\ref{1_55}), (\ref{2_55}) the equations

\begin{equation}
\label{1_59}
\bigl(a_{ij}(\xxx)\f^{(1)}_{x_i}\bigl)_{x_j}=\d_{\aaa}(\xxx)-\d_{\bbb}(\xxx)
\end{equation}

\begin{equation}
\label{2_59}
\bigl(a_{ij}(\xxx)\f^{(2)}_{x_i}\bigl)_{x_j}=\d_{\ccc}(\xxx)-\d_{\ddd}(\xxx)
\end{equation}
still with zero boundary conditions. The reciprocity condition (\ref{2_58}) holds. However, this is not the case we have in mind since with the equations (\ref{1_59}), (\ref{2_59}) the current injected in $\aaa$ is, in general, different from the current extracted from $\bbb$.
\end{remark}

\bibliographystyle{amsplain}

\end{document}